\documentclass[11pt, a4paper]{article}
\usepackage{graphicx}
\usepackage{here}
\usepackage{amsmath}
\usepackage{amsfonts}
\usepackage{amssymb}
\usepackage{amsthm}
\usepackage{mathrsfs}
\usepackage{ascmac}
\usepackage[normalem]{ulem}
\usepackage{latexsym}
\usepackage{enumerate}
\usepackage{tikz}
\usetikzlibrary{positioning}
\newtheorem{Thm}{Theorem}

\newtheorem{Lem}{Lemma}
\newtheorem{Prob}{Problem}
\newtheorem*{Claim}{Claim}

\title{Making Bidirected Graphs Strongly Connected}
\author{Tatsuya Matsuoka\thanks{The University of Tokyo, Japan ({\tt tatsuya\_matsuoka@mist.i.u-tokyo.ac.jp}).} \and Shun Sato\thanks{The University of Tokyo, Japan ({\tt shun\_sato@mist.i.u-tokyo.ac.jp}).}
}
\date{September, 2017}
\begin{document}
\maketitle

\begin{abstract}
We consider problems to make a given bidirected graph strongly connected with minimum cardinality of additional signs or additional arcs.
For the former problem, we show the minimum number of additional signs and give a linear-time algorithm for finding an optimal solution.
For the latter problem, we give a linear-time algorithm for finding a feasible solution whose size is equal to the obvious lower bound or more than that by one.
\end{abstract}

\section{Introduction}
\label{intro}
Problems to make a given graph (strongly) connected are well-investigated.
The minimum number of additional edges to make a given undirected graph connected and that of additional arcs to make a given directed graph strongly connected~\cite{ET1976} are well-known.

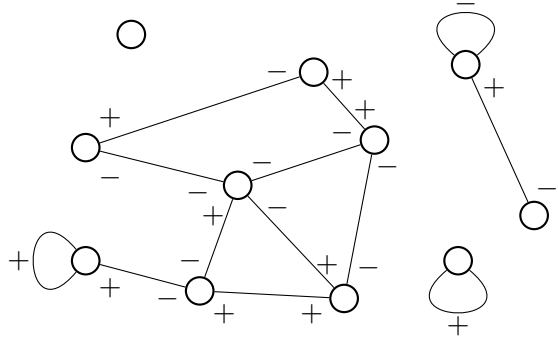
\begin{figure}[htb]
\centering
\begin{tikzpicture}[xscale = 1,yscale = 1]
\tikzset{vertex/.style={fill = white, draw, circle,thick}};
% vertices
\node[vertex] (v1) at (-1.4,2) {};
\node[vertex] (v2) at (-2,0.5) {};
\node[vertex] (v3) at (-2,-1) {};
\node[vertex] (v4) at (1,1.5) {};
\node[vertex] (v5) at (0,0) {};
\node[vertex] (v6) at (-0.5,-1.4) {};
\node[vertex] (v7) at (1.8,0.6) {};
\node[vertex] (v8) at (1.4,-1.5) {};
\node[vertex] (v9) at (2.9,-1) {};
\node[vertex] (v10) at (3,1.6) {};
\node[vertex] (v11) at (3.9,-0.4) {};
% links
\foreach \u / \v in {v2/v4,v2/v5,v3/v6,v4/v7,v5/v7,v5/v6,v5/v8,v6/v8,v7/v8,v10/v11}
\draw (\u) to (\v);
% self loops
\draw (v3) to [out=130,in=230,looseness=10] (v3);
\draw (v9) to [out=220,in=320,looseness=10] (v9);
\draw (v10) to [out=40,in=140,looseness=10] (v10);
% signs
\node [above right= 0 and -0.1 of v2] {$+$};
\node [below right= 0 and -0.1 of v2] {$-$};
\node [below right= -0.05 and -0.1 of v3] {$+$};
\node [left= 0.4 of v3] {$+$};
\node [left= 0 of v4] {$-$};
\node [below right= -0.3 and -0.05 of v4] {$+$};
\node [above right= -0.1 and -0.1 of v5] {$-$};
\node [below right= -0.1 and 0.1 of v5] {$-$};
\node [below left= -0.3 and 0.1 of v5] {$-$};
\node [below left= 0 and -0.1 of v5] {$+$};
\node [below right= -0.1 and -0.1 of v6] {$+$};
\node [above left= 0 and -0.3 of v6] {$-$};
\node [below left= -0.3 and 0 of v6] {$-$};
\node [above left= 0 and -0.3 of v7] {$+$};
\node [above left= -0.3 and 0 of v7] {$-$};
\node [below right= -0.05 and -0.25 of v7] {$-$};
\node [above right= 0 and -0.1 of v8] {$-$};
\node [above left= 0.05 and -0.2 of v8] {$+$};
\node [below left= -0.2 and 0 of v8] {$+$};
\node [below = 0.4 of v9] {$+$};
\node [above = 0.35 of v10] {$-$};
\node [below right= -0.1 and -0.05 of v10] {$+$};
\node [above right= -0.05 and -0.25 of v11] {$-$};
\end{tikzpicture}
\caption{Bidirected Graph.}
\label{FigBidirectedGraph}
\end{figure}
The concept of bidirected graphs (Figure~\ref{FigBidirectedGraph}; the precise definition will be given later in Section 2) was introduced by Edmonds and Johnson~\cite{EJ1970}.
It is a common generalization of undirected graphs and directed graphs.
For bidirected graphs, Ando, Fujishige and Nemoto~\cite{AFN1996} defined the notion of strong connectivity and gave a linear-time algorithm for the strongly connected component decomposition.
However, problems to make a given bidirected graph strongly connected have not been formulated.

In this paper, we consider problems to make a given bidirected graph strongly connected with minimum cardinality of additional signs or additional arcs.

\subsection{Related Works}
It is obvious that the minimum number of additional edges to make a given undirected graph connected is fewer than the number of connected components of a given graph by one.
Eswaran and Tarjan~\cite{ET1976} gave the minimum number of additional arcs to make a given directed graph strongly connected and that of additional edges to make a given undirected graph \emph{bridge-connected} (\emph{2-edge-connected}) or \emph{biconnected} (\emph{2-vertex-connected}).
Linear-time algorithms for finding an optimal solution of these problems are also given in \cite{ET1976}.
Note that they defined an operation called ``condensation'' which transforms a general directed graph to an acyclic directed graph.
We can focus on the acyclic case since for this problem we can obtain a solution of the original problem by solving the problem on the condensed graph.
For a directed graph $G=(V, A)$, $v\in V$ is a \emph{source} if $\delta (v)\geq 1,\ \rho (v)=0$, a \emph{sink} if $\rho (v)\geq 1,\ \delta (v)=0$ and an \emph{isolated vertex} if $\rho (v)=\delta (v)=0$ (in directed graphs, $\delta$ and $\rho$ denote the out-degree and in-degree functions, respectively).

\begin{Thm}[Eswaran--Tarjan~\cite{ET1976}]
Let $G=(V, A)$ be an acyclic directed graph with the set $S\subseteq V$ of sources, the set $T\subseteq V$ of sinks and the set $Q\subseteq V$ of isolated vertices $(|S|+|T|+|Q|>1)$.
Then the minimum number of additional arcs to make the given graph strongly connected is $\max\{ |S|, |T|\} +|Q|$.
\end{Thm}

For an undirected graph $G=(V, E)$, $v\in V$ is called a \emph{pendant} if $\delta (v)=1$ and $V'\subseteq V$ is called a \emph{pendant block} if it is a 2-vertex-connected component and it contains exactly one \emph{cutnode} (for undirected graphs, $\delta$ denotes the degree function).
Note that $v\in V$ is a cutnode if the original graph is connected and the graph induced by $V\setminus \{ v\} $ is disconnected.
Similarly, $V'\subseteq V$ is called an \emph{isolated block} if it is a 2-vertex-connected component and it contains no cutnode.

\begin{Thm}[Eswaran--Tarjan~\cite{ET1976}]
Let $G=(V, E)$ be an undirected graph with the set $P\subseteq V$ of pendants and the set $Q\subseteq V$ of isolated vertices $(|P|+|Q|>1)$.
Then the minimum number of additional edges to make the given graph 2-edge-connected is $\lceil |P|/2\rceil +|Q|$.
\end{Thm}
\begin{Thm}[Eswaran--Tarjan~\cite{ET1976}]
Let $G=(V, E)$ be an undirected graph with the set $\mathcal{P}\subseteq 2^V$ of pendant blocks and the set $\mathcal{Q}\subseteq 2^V$ of isolated blocks $(|\mathcal{P}|+|\mathcal{Q}|>1)$. Then the minimum number of additional edges to make the given graph 2-vertex-connected is $\max \left\{ d-1, \lceil |\mathcal{P}|/2\rceil +|\mathcal{Q}|\right\}$.
Here, 
\begin{align*}
d:=&\max \{ \# (\mbox{2-vertex-connected components containing }v)\mid v\in V \} \\
&+\# (\mbox{connected components})-1.
\end{align*}
\end{Thm}

On the other hand, problems on bidirected graphs also have been considered in the literature.
Ando, Fujishige and Nemoto~\cite{AFN1996} gave a linear-time algorithm for strongly connected component decomposition of bidirected graphs.
This algorithm is made use of for the block triangularization of skew-symmetric matrices~\cite{Iwata1998}.
Bidirected graphs are also used in the field of computational biology~\cite{MB2009,MGMB2007,Yasuda2015}.

The strongly connected component decomposition of a bidirected graph~\cite{AFN1996} is obtained by the ordinary strongly connected component decomposition of the associated directed graph, \emph{skew-symmetric graph}, which will be used in Section~3.
As pointed out in \cite{AFN1996}, the same graph is used by Zaslavsky~\cite{Zaslavsky1991} for the study of \emph{signed graphs}~\cite{Harary19531954}.
The notion of skew-symmetric graphs is defined first by Tutte~\cite{Tutte1967} with the name ``antisymmetrical digraphs'' independent from bidirected graphs.
There are also various problems on skew-symmetric graphs, and they have been intensively studied~\cite{GK1995,GK1996,GK2004}.
Study on bisubmodular polyhedra also made use of this skew-symmetric graph~\cite{AF1996} with the name ``exchangeability graph.''

\subsection{Our Contribution}
In this paper, we formulate problems to make a given bidirected graph strongly connected with minimum cardinality of additional signs or additional arcs.
Since self-loops have influence on the strong connectivity on bidirected graphs, 
these two problems arise depending on how to treat self-loops.

We first define the procedure called ``condensation'' on bidirected graphs.
We can reduce general cases to acyclic cases by this operation for the above two problem settings.
This can be done by using the strongly connected component decomposition algorithm for bidirected graphs devised  by Ando, Fujishige and Nemoto~\cite{AFN1996}.
This is similar to the fact that the condensation on directed graphs is done by using strongly connected component decomposition of directed graphs~\cite[Lemma~1]{ET1976}. 
However, since there are signs on each arc in bidirected graphs, we must define the appropriate signs for each arc on the condensed bidirected graph.

We discuss the two versions of the problems on bidirected graphs.
For the problem on signs, the obvious lower bound can be obtained from the necessity for connectivity of the underlying graph and a condition on signs around each vertex.
We show that this lower bound can be achieved for any acyclic bidirected graph and give a linear-time algorithm for finding an optimal solution.
For the problem on arcs, we give a linear-time algorithm for finding a feasible solution whose size is equal to the obvious lower bound or more than that by one.

\subsection{Organization}
The organization of the rest of this paper is as follows.
We give definitions and notation in Section 2.
In Section 3, we give two problem settings dealt with in this paper and devise the condensation operation on bidirected graphs, which reduces a general case to an acyclic case.
These two problems are discussed in Sections 4 and 5, respectively.
Section 6 is devoted to concluding remarks involving other problem settings.

\section{Preliminaries}
In this section, we introduce definitions and notation used in this paper.
Definitions in this section mainly refer Ando and Fujishige~\cite{AF1994} and Ando, Fujishige and Nemoto~\cite{AFN1996}.

A \emph{bidirected graph} is a triplet of a vertex set $V$, an arc set $A$ and a boundary operator $\partial : A\to 3^V (:=\{ (X, Y)\mid X, Y\subseteq V, X\cap Y=\emptyset \} )$ such that $\partial a=(X_a, Y_a)$ satisfies $1\leq |X_a|+|Y_a|\leq 2$ for each $a\in A$.
We use the notation $|\partial a|:=|X_a|+|Y_a|$.
Let $\partial^+: A\to 2^V$ and $\partial^-: A\to 2^V$ denote the operators with $\partial^+a=X_a$ and $\partial^-a=Y_a$.
This can be regarded that the signs are put on endpoints of \emph{links} or on \emph{self-loops} by $\partial^+$ and $\partial^-$ (here we call an arc a link if it connects two distinct vertices).
In other words, $\partial^+a$ and $\partial^-a$ are the sets of endpoints of $a$ with the signs ``$+$'' and ``$-$'', respectively.
We call an arc $a$ with $\partial a=(\{ v\} , \emptyset )$ a plus-loop at $v$ and $a$ with $\partial a=(\emptyset , \{ v\} )$ a minus-loop at $v$.

For simplicity, we define some other notation.
Let $\bar{\partial}: A\to 2^V$ denote the operator with $a\mapsto \partial^+a\cup \partial^-a$ for each $a\in A$.
For a bidirected graph $G=(V, A; \partial )$, let $\bar{G}=(V, A)$ be the undirected graph omitting the signs of $G$ (the \emph{underlying graph} of $G$).
We write as $a=(u, v)$ if $\bar{\partial}a=\{ u, v\}$.
Let us define a sign operator $\pi : \left\{ (a, u)\mid a\in A, u\in \bar{\partial}a\right\} \to \{ +, -\}$ as $\pi (a, u)=+$ if $u\in \partial^+a$ and $\pi (a, u)=-$ if $u\in \partial^-a$.
Let ``$(u, v)$ with $(\pi_1, \pi_2 )$'' $(\pi _1, \pi _2\in \{ +, -\} )$ denote an arc $a=(u, v)$ with $\pi (a, u)=\pi_1 $ and $\pi (a, v)=\pi_2 $.

An arc $a\in A$ is said to be \emph{positively (resp. negatively) incident to} $v$ if $v\in \partial^+a$ $(\mbox{resp. }v\in \partial^-a)$.
Arcs $a\in A$ and $a'\in A$ are said to be \emph{oppositely incident to} $v$ if $a$ is positively (resp. negatively) incident to $v$ and $a'$ is negatively (resp. positively) incident to $v$.

An alternating sequence of vertices and arcs $(v_0, a_1, v_1, a_2, \ldots , a_l, v_l)\ (l\geq 1)$ is called a \emph{path} if $a_i$ and $a_{i+1}$ are oppositely incident to $v_i$ $(i=1, 2, \ldots , l-1)$, $a_1$ is incident to $v_0$ and $a_l$ is incident to $v_l$.
This is called ($\pi (a_1, v_0)$, $\pi (a_l, v_l)$)-path from $v_0$ to $v_l$.
A path with $v_0=v_l$ is called a \emph{cycle} with a \emph{root} $v_0(=v_l)$.
If $a_l$ and $a_1$ are oppositely incident to $v_0$ and it includes distinct vertices, we call it a \emph{proper} cycle.
A cycle which is not proper is called an \emph{improper} cycle.
An improper cycle with $\pi (a_1, v_0)=\pi (a_l, v_l(=v_0))=+ \ (\mathrm{resp.} -)$ is called a \emph{positive (resp. negative) improper cycle}.
If a graph does not contain a proper cycle, we call it an \emph{acyclic} graph (Note that this definition is different from that of ``strongly acyclic'' or ``weakly (node- or edge-) acyclic'' in \cite{Babenko2006}).

For a bidirected graph $G=(V, A; \partial )$, two vertices $v, v'\in V$ are called \emph{strongly connected} if $G$ contains two paths $(v, a_1^1, v_1^1, a_2^1, \ldots , a_{l_1}^1, v')$ and $(v, a_1^2, v_1^2, a_2^2, \ldots , a_{l_2}^2, v')$ such that $a_1^1$ and $a_1^2$  are oppositely connected to $v$ and that $a_{l_1}^1$ and $a_{l_2}^2$ are oppositely connected to $v'$.
Note that these two paths need not to be vertex-disjoint.
A binary relation on $V$ can be defined by this strong connectivity: $v\sim v'$ if $v$ and $v'$ are strongly connected.
By assuming that $v\sim v$ for all $v\in V$, we obtain the equivalence relation $\sim$ on $V$.
Each equivalence class of $V$ on $\sim$ is called \emph{strongly connected component} and $G$ is called \emph{strongly connected} if $G$ has only one strongly connected component.

A vertex $v\in V$ is called \emph{inconsistent} if there exist improper cycles with root $v$ $C_1=(v, a_1^1, v_1^1, a_2^1, \ldots , a_{l_1}^1, v)$ and $C_2=(v, a_1^2, v_1^2, a_2^2, \ldots , a_{l_2}^2, v)$ such that $a_{l_2}^2$ and $a_1^1$ are oppositely incident to $v$.
It is stated in~\cite{AFN1996} that if $u$ and $v$ are strongly connected and $u$ is inconsistent, then $v$ is also inconsistent.
Thus, the notion of inconsistency can also be naturally defined on strongly connected components.

\section{Settings and the Condensation Operation}
In this section, we introduce the problem settings we tackle in this paper and explain the operation called condensation.

\subsection{Problem Settings}
We deal with the problems of the following type.
\begin{Prob}
Let $G=(V, A; \partial )$ be a bidirected graph.
Find additional arcs $A'$ and a boundary operator $\partial' : A\cup A'\to 3^V (\partial' a=\partial a\ (\forall a\in A))$ minimizing $F(A', \partial'):=\sum_{a'\in A'}f(\partial' a')$ $(f: \{ (X, Y)\mid X, Y\subseteq V, X\cap Y=\emptyset , 1\leq |X|+|Y|\leq 2\} \to \mathbb{R})$ such that $G':=(V, A\cup A'; \partial' )$ is a strongly connected bidirected graph.
\label{ProbGeneral}
\end{Prob}
Note that Problem~\ref{ProbGeneral} is NP-hard in general.
This can easily be shown by following the argument in the proof of~\cite[Theorem~1]{ET1976} as follows: 
we show this by reducing the following directed Hamiltonian cycle problem to Problem~\ref{ProbGeneral} with a certain function~$f$.
\begin{Prob}[Directed Hamiltonian Cycle Problem]
Let $D=(V, A)$ be a directed graph.
Find a directed Hamiltonian cycle in $D$.
\label{ProbDHC}
\end{Prob}
Set $f(\partial' a')=1$ if $a'=(v_1, v_2)$ with $(+, -)$ and there exists $a=(v_1, v_2)$ in $D$ and set $f(\partial' a')=2$ for any other possible arc $a'$.
There exists a solution of Problem~\ref{ProbGeneral} with respect to $ G = (V, \emptyset; \partial) $ satisfying $F(A', \partial' )=|V|$ if and only if there exists a solution of Problem~\ref{ProbDHC}.
Since Problem~\ref{ProbDHC} is NP-complete~\cite{Karp1972}, Problem~\ref{ProbGeneral} is NP-hard.

For the problem on undirected graphs or directed graphs similar to Problem~\ref{ProbGeneral}, it is natural to define $F(A', \partial' ):=|A'|$, i.e., minimization of the cardinality of additional edge (or arc) set.
For bidirected graphs, however, there are two reasonable candidates of $F(A', \partial' )$, i.e., $\sum_{a'\in A'}|\partial' a'|$ and $|A'|$.
In other words, $f(\partial' a'):=|\partial' a'|$ in the former setting and $f(\partial' a'):=1$ in the latter setting.
The former means the minimization of the number of the additional signs on arcs and the latter means that of arcs themselves.
In other words, the cost of a link is twice higher than that of a self-loop for the former problem and is the same for the latter problem.
These natural two problems arise because self-loops have influence on strong connectivity in bidirected graphs 
(see, e.g., Figure~\ref{FigEx}). 
Note that self-loops do not have any influence on the structure of (strong) connectivity for undirected graphs or directed graphs.

\subsection{Reduction to Acyclic Case}
We present a technique for reducing general cases to acyclic cases for Problem~\ref{ProbGeneral} with respect to $F(A', \partial' )=\sum_{a'\in A'}|\partial' a'|$ or $F(A', \partial' )=|A'|$.

For directed graphs, Eswaran and Tarjan [6] first condense the given directed graph to focus on acyclic cases. 
There, the condensed graph $\tilde{G} = (\tilde{V}, \tilde{A} )$ is obtained from the strongly connected component decomposition $C_1,C_2,\dots ,C_k$ of the original graph $G = (V,A)$, where 
\begin{align*}
\tilde{V} &:=\{ v_{C_1}, v_{C_2}, \ldots , v_{C_k}\},\\
\tilde{A} &:=\left\{ (v_{C_j}, v_{C_k})\mid \exists v\in V(C_j), \exists v'\in V(C_k)\mbox{ s.t. }(v, v')\in A\right\}.
\end{align*}

For bidirected graphs, we can utilize the linear-time algorithm for strongly connected component decomposition devised by Ando, Fujishige and Nemoto~\cite{AFN1996}. 
Precisely speaking, in order to appropriately define signs in the condensed graph, 
we use the strongly connected component decomposition of the associated skew-symmetric graph, 
which corresponds to the strongly connected component decomposition of the original bidirected graph $ G = (V, A; \partial) $~\cite[Corollary~5.4]{AFN1996}. 

In Algorithm CONDENSE($G$) described below, Steps~1--3 based on the steps of the strongly connected component decomposition algorithm of~\cite{AFN1996}. 
\newline

\noindent \textbf{Algorithm} CONDENSE($G$)
\begin{description}
\item[Step 1] Construct the associated skew-symmetric graph $G^{\pm}=(V^+\cup V^-, A^{\pm}) $, where $V^+$ and $V^-$ are copies of $V$ ($v^+\in  V^+$ and $v^-\in V^-$ denote the copy of $v\in V$) and $A^{\pm}$ are defined by 
\[
A^{\pm}=\{ (v^{\pi (a,v)}, w^{-\pi (a,w)})\mid a\in A, \bar{\partial}a=\{ v,w\} \}.
\]
Note that $v$ can be equal to $w$.
Here, $-\pi$ is equal to $-$ (resp. $+$) if $\pi =+$ (resp. $-$).
\item[Step 2] Decompose $G^{\pm}$ into strongly connected components $G^{\pm}_j=(U^{\pm}_j, B^{\pm}_j)$ $( j\in J)$. 
\item[Step 3] For each $j\in J$, define
\[
U_j=\{ v\in V\mid v^+\in U^{\pm}_j\} \cup \{ v\in V\mid v^-\in U^{\pm}_j\}.
\]
Then, define $W_i$ $(i\in I)$ be the distinct members of $U_j$ $(j\in J)$ 
and partition $I$ into $I_1$ and $I_2$ so that $W_i$ appears twice (resp. once) in the family $\{ U_j\mid j\in J\}$ for each $i\in I_1$ (resp. $I_2$).
\item[Step 4] For each $i\in I_1$, let $U_j^{\pm}$ be one of two strongly connected components corresponding to $W_i$.
If $U_j^{\pm}$ includes both an element in $V^+$ and an element in $V^-$, then for each $v^-\in V^-\cap U_j^{\pm}$ swap this for the counterpart.
\item[Step 5] Make a skew-symmetric graph $\hat{G}^{\pm}=(\hat{V}^+\cup \hat{V}^-, \hat{A})$ from $G^{\pm}$ as follows.
Let $\hat{v}_i^+$ (resp. $\hat{v}_i^-$) be a representative of $W_i$.
Let $\hat{V}^+:=\{ \hat{v}_i^+\mid i\in I\}$ and $\hat{V}^-:=\{ \hat{v}_i^-\mid i\in I\}$.
By using the map $\alpha : V^+\cup V^-\to \hat{V}^+\cup \hat{V}^-$ defined by
\[
\alpha (v^{\pi})=\hat{v}_i^{\pi}\quad (v\in W_i, \pi \in \{ +, -\} ),
\]
the arc set $\hat{A}$ is defined by
\[
\hat{A}=\{ (\alpha (v), \alpha (w))\mid (v, w)\in A^{\pm}\ (\alpha (v)\neq \alpha (w))\} .
\]
\item[Step 6] Return the bidirected graph $\tilde{G}$ corresponding to the skew-symmetric graph~$\hat{G}^{\pm}$.
\end{description}
Note that the strongly connected component $W_i$ is inconsistent if and only if $i\in I_2$ (see, Corollary 5.4 of \cite{AFN1996}).

From a feasible solution of the problem of minimization on signs or arcs for $\tilde{G}$, we can obtain a feasible solution for $G$ with the same value for the function $F$.
Conversely, from any feasible arc set for the original problem on $G$ we can obtain a solution for the problem on $\tilde{G}$ whose cost is less than or equal to the original cost.
These hold since a condensed graph of any strongly connected bidirected graph is strongly connected and each link in the obtained solution graph corresponds with links in the original solution graph.
Thus validity of the above condensation holds.

\section{Minimization on Signs}
In this section, we deal with Problem~\ref{ProbGeneral} with $F(A', \partial' )=\sum_{a'\in A'}|\partial' a'|$.

We first give some definitions on bidirected graphs.
Let $\gamma (=\gamma (G))$ denote the number of connected components in the underlying graph $\bar{G}$ of $G$.
A vertex $v\in V$ is called a \emph{source} (resp. a \emph{sink}) if $v$ is included in a connected component in $\bar{G}$ which has more than one vertices and any $a\in A$ connected to $v$ in $G$ is positively (resp. negatively) incident to $v$.
The set of sources is denoted by $S(=S(G))$ and that of sinks is denoted by $T(=T(G))$.
A vertex $v\in V$ is called an \emph{isolated vertex} if there exists no arc connected to $v$.
The set of isolated vertices is denoted by $Q(=Q(G))$.
A vertex $v\in V$ is called a \emph{pseudo-isolated vertex} if $\{ v\}$ is the connected component with only one vertex in $\bar{G}$ and there exists a self-loop at $v$.
The set of pseudo-isolated vertices is denoted by $Q'(=Q'(G))$.

When adding an arc $a=(u, v)$ to a bidirected graph $G$, we write ``with proper signs'' if signs on $a$ are as follows: $\pi (a, u)$ is equal to $+$ if $\{ a\in A\mid u\in \partial^+a\}$ is empty for the current bidirected graph and $\pi (a, u)$ is equal to $-$ otherwise.
The sign $\pi (a, v)$ is determined in the same way.

Now, let us consider Problem~\ref{ProbGeneral} with respect to the number of additional signs on an acyclic bidirected graph $G=(V, A; \partial )$.
Since a bidirected graph is strongly connected only if its underlying graph is connected, the value of the objective function for a feasible solution must be greater than or equal to $2(\gamma -1)$.
On the other hand, a bidirected graph with $|V|>1$ is strongly connected only if there are no sources, sinks, isolated vertices and pseudo-isolated vertices.
Therefore, the number of additional signs to make a bidirected graph strongly connected is greater than or equal to $|S|+|T|+|Q'|+2|Q|$.
Summing up, we obtain the lower bound $\max \{ 2(\gamma -1), |S|+|T|+|Q'|+2|Q|\}$.
Actually, this lower bound can be achieved.
\begin{Thm}
Let $G=(V, A; \partial )$ be an acyclic bidirected graph with $|V|>1$.
Then the minimum number of $\sum_{a'\in A'}|\partial' a'|$ such that $G'=(V, A\cup A'; \partial' )$ is a strongly connected bidirected graph is $\max \{ 2(\gamma -1), |S|+|T|+|Q'|+2|Q|\}$.
\label{ThmSign}
\end{Thm}
We now describe an algorithm for constructing an optimal solution (whose size is equal to the lower bound).
Let $C_1^1, C_2^1, \ldots , C_{k_1}^1, C_1^2, C_2^2, \ldots , C_{k_2}^2, \ldots , C_1^K, C_2^K, \ldots , C_{k_K}^K$ be the distinct vertex sets of connected components of $\bar{G}$ such that each $C_i^j$ contains just $j$ elements of $S\cup T\cup Q'\cup Q$.
Note that $\sum_{i=1}^Kk_i=\gamma$ and $\sum_{i=1}^Kik_i=|S|+|T|+|Q'|+|Q|$.
\newline

\noindent \textbf{Algorithm} ADDITIONAL SIGNS($G$)
\begin{description}
\item[Step 1] Let $A':=\emptyset$.
\item[Step 2] 
Let $u_1, u_2, \ldots , u_{L_1} \ (L_1:=k_1-|Q|)$ be the elements of $\left( \bigcup_{i=1}^{k_1}C_i^1 \right) \cap (S\cup T\cup Q')$.
If $L_1=\gamma =1$, add a self-loop at $u_1$ to $A'$ with a proper sign and go to Step 6.
If $L_1=\gamma >1$, add $\{ (u_1, u_i)\mid 2\leq i \leq L_1\}$ to $A'$ with proper signs and go to Step 6. 
\item[Step 3]
Let $\mathcal{C}=\left\{ C_1^2, C_2^2, \ldots , C_{k_2}^2, \ldots , C_1^K, C_2^K, \ldots , C_{k_K}^K \right\}$.
For each $C\in \mathcal{C}$, pick up two distinct elements of $C\cap (S\cup T)$ and label them as $l_i, r_i \ (i=1, 2, \ldots , |\mathcal{C}|)$.
Label the rest of the elements of $\bigcup_{C\in \mathcal{C}}C\cap (S\cup T)$ as $w_1, w_2, \ldots , w_{L_2}$ with $L_2 :=\sum_{i=3}^{K}(i-2)k_{i}$.
Add $\left\{ (u_i, w_i)\mid 1\leq i\leq \min\{ L_1, L_2 \} \right\}$ to $A'$ with proper signs.
\item[Step 4] 
Let $q_1, \ldots , q_{|Q|}$ be the elements of $Q$ and define $l_{|\mathcal{C}|+i}=r_{|\mathcal{C}|+i}=q_i$ for $i = 1, \ldots , |Q|$. 
Add $\left\{ (r_i, l_{i+1})\mid 1\leq i<|\mathcal{C}|+|Q|\right\}$ to $A'$ with proper signs. 
\item[Step 5] Compare $L_1$ with $L_2$.
\begin{description}
\item[Step 5-1] If $L_2 \leq L_1-2$, add $(u_{L_2+1}, l_1)$ and $\left\{ (u_i, r_{|\mathcal{C}|+|Q|})\mid L_2+1<i\leq L_1\right\}$ to $A'$ with proper signs.
\item[Step 5-2] If $L_2 =L_1-1$, add $(u_{L_1}, l_1)$ and a self-loop at $r_{|\mathcal{C}|+ |Q|}$ to $A'$ with proper signs.
\item[Step 5-3] If $L_2 \geq L_1$, add self-loops at $l_1$, $r_{|\mathcal{C}|+|Q|}$ and $w_i$ for $i=L_1+1, L_1+2, \ldots , L_2$ to $A'$ with proper signs.
\end{description}
\item[Step 6] Return $G'=(V, A\cup A'; \partial' )$.
\end{description}
Steps 3 and 4 are like as in Figure~\ref{FigAlgo1}.
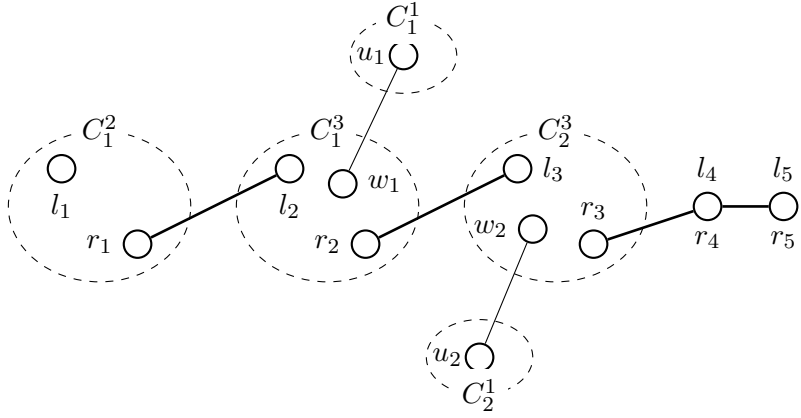
\begin{figure}[htb]
\centering
\begin{tikzpicture}[xscale = 1,yscale = 1]
\tikzset{vertex/.style={fill = white, draw, circle,thick}};
% \mathcal{C} 1
\node[vertex] (l1) at (-5.5,0.5) {};
\node[below = 0 of l1] {$l_1$};
\node[vertex] (r1) at (-4.5,-0.5) {};
\node[left = 0 of r1] {$r_1$};
\draw[dashed] (-5,0) circle [x radius=1.2, y radius=1];
\node[rectangle,fill=white] at (-5,1) {$C^2_1$};
% \mathcal{C} 2
\node[vertex] (l2) at (-2.5,0.5) {};
\node[below = 0 of l2] {$l_2$};
\node[vertex] (r2) at (-1.5,-0.5) {};
\node[left = 0 of r2] {$r_2$};
\node[vertex] (w1) at (-1.8,0.3) {};
\node[right = 0 of w1] {$w_1$};
\draw[dashed] (-2,0) circle [x radius=1.2, y radius=1];
\node[rectangle,fill=white] at (-2,1) {$C^3_1$};
% \mathcal{C} 3
\node[vertex] (l3) at (0.5,0.5) {};
\node[right = 0 of l3] {$l_3$};
\node[vertex] (r3) at (1.5,-0.5) {};
\node[above = 0 of r3] {$r_3$};
\node[vertex] (w2) at (0.7,-0.3) {};
\node[left = 0 of w2] {$w_2$};
\draw[dashed] (1,0) circle [x radius=1.2, y radius=1];
\node[rectangle,fill=white] at (1,1) {$C^3_2$};
% isolated
\node[vertex] (vi1) at (3,0) {};
\node[above = 0 of vi1] {$l_4$};
\node[below = 0 of vi1] {$r_4$};
\node[vertex] (vi2) at (4,0) {};
\node[above = 0 of vi2] {$l_5$};
\node[below = 0 of vi2] {$r_5$};
% pseudo isolated 1
\node[vertex] (vp1) at (-1,2) {};
\node[left = -0.1 of vp1] {$u_1$};
\draw[dashed] (-1,2) circle [x radius=0.7, y radius=0.5];
\node[rectangle,fill=white] at (-1,2.5) {$C^1_1$};
% pseudo isolated 2
\node[vertex] (vp2) at (0,-2) {};
\node[left = -0.1 of vp2] {$u_2$};
\draw[dashed] (0,-2) circle [x radius=0.7, y radius=0.5];
\node[rectangle,fill=white] at (0,-2.5) {$C^1_2$};
% additional arcs
\foreach \u / \v in {r1/l2,r2/l3,r3/vi1,vi1/vi2}
\draw[line width=1pt] (\u) to (\v);
\foreach \u / \v in {vp1/w1,vp2/w2}
\draw[thin] (\u) to (\v);
\end{tikzpicture}
\caption{Additional arcs in Steps 3 and 4 in the proposed algorithm: thin and bold lines represent the additional arcs in Steps 3 and 4, respectively (Arcs in each connected component are omitted).}
\label{FigAlgo1}
\end{figure}

The above algorithm returns an optimal solution in linear time.
This is confirmed by the following two lemmas.
\begin{Lem}
The output of the above algorithm is strongly connected.
\end{Lem}
This can be confirmed by the following claims when $\gamma >L_1$.
(It can be shown more easily when $\gamma =L_1$.)
\begin{Claim}
The vertex set $\{ l_i, r_i\mid 1\leq i\leq |\mathcal{C}|\} \cup Q$ is strongly connected.
\end{Claim}
\begin{proof}
For each $C\in \mathcal{C}$, there exists a path between $l$ and $r$ (vertices chosen as $l_i$ and $r_i$).
This can be shown by the facts that $l$ and $r$ are connected in the underlying graph and that every vertex in $C$ has both plus and minus signs around it.
Since $l_1$ and $r_{|\mathcal{C}|+|Q|}$ have a self-loop or an improper cycle, the claim holds.
\end{proof}
\begin{Claim}
Each vertex in $\{ l_i, r_i\mid 1\leq i\leq |\mathcal{C}|\} \cup Q$ is inconsistent.
\end{Claim}
\begin{proof}
Suppose there is a $(+, +)$-path between $l_1$ and $r_{|\mathcal{C}|+|Q|}$. (The other case can be treated in the similar way.)
By the algorithm, there are a negative improper cycle rooted at $l_1$ and that rooted at $r_{|\mathcal{C}|+|Q|}$.
Thus due to the above $(+, +)$-path, there is a positive improper cycle rooted at $l_1$.
Therefore $l_1$ is inconsistent and hence the claim holds.
\end{proof}
\begin{Claim}
For each vertex $v\in V\setminus (\{ l_i, r_i\mid 1\leq i\leq |\mathcal{C}|\} \cup Q)$, there exist $v_1^*, v_2^*\in \{ l_i, r_i\mid 1\leq i\leq |\mathcal{C}|\} \cup Q$ (not necessarily distinct), a path $P_1$ between $v$ and $v_1^*$ and a path $P_2$ between $v$ and $v_2^*$ such that end arcs of $P_1$ and $P_2$ connected to $v$ are oppositely incident.
\end{Claim}
\begin{proof}
By the algorithm ADDITIONAL SIGNS, each vertex in the resultant graph has both plus and minus signs around it.
Fix a vertex $v\in V\setminus (\{ l_i, r_i\mid 1\leq i\leq |\mathcal{C}|\} \cup Q)$.
By the above two claims, there is an inconsistent strongly connected component including $\{ l_i, r_i\mid 1\leq i\leq |\mathcal{C}|\} \cup Q$.
Since the underlying graph is connected and each vertex has both signs around it, we can reach this component from $v$ regardless of the initial sign.
Thus we can obtain $v_1^*, v_2^*, P_1$ and $P_2$.
\end{proof}
\begin{Claim}
If a vertex set $V'$ contains inconsistent vertices $v_1$ and $v_2$, and for each $v\in V'\setminus \{ v_1, v_2\}$ there are $(v, v_1)$-path and $(v, v_2)$-path with the opposite starting sign around $v$, then the whole $V'$ is strongly connected.
\end{Claim}
\begin{proof}
It holds since there exists a proper cycle including the above paths which passes $v_1$ and $v_2$ twice and $v$.
\end{proof}
Next, we check the number of additional signs.
\begin{Lem}
The number of additional signs is equal to $\max \{ 2(\gamma -1), |S|+|T|+|Q'|+2|Q|\}$.
\end{Lem}
\begin{proof}
If $L_1=\gamma =1$, add only one self-loop thus $1=\max \{ 0, 1\}$.
If $L_1=\gamma >1$, add $\gamma -1$ links thus $2(\gamma -1)=\max \{ 2(\gamma -1), \gamma \}$.

Otherwise, we go to Step~3 and add $\min \{ L_1, L_2\}$ links.
Next, we add $|\mathcal{C}|+|Q|-1$ links at Step~4.

At Step~5, we consider three cases.
It should be noted that 
$L_2 \leq L_1 - 2$ holds if and only if $2(\gamma -1)\geq |S|+|T|+|Q'|+2|Q|$ holds due to the following relation:
\begin{align*}
L_2-(L_1-2)&=\sum_{i=2}^{K}(i-2)k_i-(k_1-|Q|)+2\\
&=\sum_{i=1}^K ik_i-2\gamma +|Q|+2\\
&=|S|+|T|+|Q'|+2|Q|-2(\gamma -1).
\end{align*}
If $L_2\leq L_1-2$, we add $L_1-L_2$ links, thus the number of additional signs is
\begin{align*}
&2\min \{ L_1, L_2\} +2(|\mathcal{C}|+|Q|-1)+2(L_1-L_2)\\
={}&2(|\mathcal{C}|+|Q|+L_1-1)\\
={}&2(|\mathcal{C}|+k_1-1)\\
={}&2(\gamma -1)\\
={}&\max \{ 2(\gamma -1), |S|+|T|+|Q'|+2|Q|\} .
\end{align*}
If $L_2=L_1-1$, we add a link and a self-loop, thus the number of additional signs is
\begin{align*}
&2\min \{ L_1, L_2\} +2(|\mathcal{C}|+|Q|-1)+3\\
={}&2\sum_{i=2}^K(i-2)k_i+2\sum_{i=2}^Kk_i+2|Q|+1\\
={}&\sum_{i=2}^K(i-2)k_i+\sum_{i=1}^Kk_i+\sum_{i=2}^Kk_i+|Q|\\
={}&\sum_{i=1}^Kik_i+|Q|\\
={}&|S|+|T|+|Q'|+2|Q|\\
={}&\max \{ 2(\gamma -1), |S|+|T|+|Q'|+2|Q|\} .
\end{align*}
Otherwise, the number of additional signs is
\begin{align*}
&2\min \{ L_1, L_2\} +2(|\mathcal{C}|+|Q|-1)+(L_2-L_1+2)\\
={}&2(|\mathcal{C}|+|Q|)+L_2+L_1\\
={}&2\sum_{i=2}^Kk_i+2|Q|+\sum_{i=2}^K(i-2)k_i+L_1\\
={}&\sum_{i=1}^Kik_i+2|Q|-k_1+L_1\\
={}&|S|+|T|+|Q'|+2|Q|\\
={}&\max \{ 2(\gamma -1), |S|+|T|+|Q'|+2|Q|\} .
\end{align*}
Therefore, the number of additional signs is equal to the obvious lower bound.
\end{proof}
Both the above algorithm and the condensation algorithm run in linear time, thus one can obtain an optimal solution in linear time for a general input bidirected graph.
\begin{Thm}
Problem~\ref{ProbGeneral} with $F(A', \partial' )=\sum_{a'\in A'}|\partial' a'|$ can be solved in linear time.
\end{Thm}

\section{Minimization on Arcs}
In this section, we deal with Problem~\ref{ProbGeneral} with $F(A', \partial' )=|A'|$.

Let $\lambda (G)$ be defined by $\lambda (G):=\max \left\{ \gamma -1, \lceil (|S|+|T|+|Q'|)/2\rceil +|Q|\right\} $.
Clearly, $\lambda (G)$ is the lower bound of Problem~\ref{ProbGeneral} with $F(A', \partial' )=|A'|$ (which can be derived as well as that for the problem on signs).
Unfortunately, however, there is a small example which cannot be made strongly connected by $\lambda (G)$ additional arcs (see Figure~\ref{FigEx}), whereas we can always achieve the lower bound when we deal with the number of additional signs as shown in the previous section.
For the original graph $G$ in Figure~\ref{FigEx} (a), we have
\[
\lambda (G)=\max \left\{ \bigg\lceil \frac{1+1+0}{2}\bigg\rceil +0\right\} =\max \left\{ 1, 0\right\} =1.
\]
Since there exist a source $s$ and a sink $t$ in $G$, 
we must add an arc $a=(s, t)$ with proper signs in order to extinguish both source and sink with one arc (see Figure~\ref{FigEx} (b)).
However, it is not strongly connected.
Actually, the minimum number of additional arcs to make $G$ strongly connected is two and one of the optimal solutions is shown in Figure~\ref{FigEx} (c).
On the other hand, there is also an graph $G$ which can be made strongly connected with $\lambda (G)$ additional arcs.
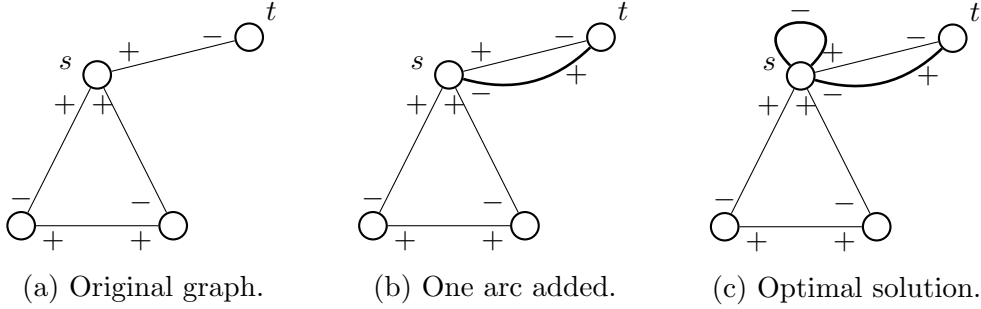
\begin{figure}[H]
\centering
\begin{minipage}{0.3\textwidth}
\centering
\begin{tikzpicture}
\tikzset{vertex/.style={fill = white, draw, circle,thick}};
% vertices
\node[vertex] (v1) at (2,2) {};
\node[vertex] (v2) at (0,1.5) {};
\node[vertex] (v3) at (-1,-0.5) {};
\node[vertex] (v4) at (1,-0.5) {};
% links
\foreach \u / \v in {v1/v2,v2/v3,v3/v4,v4/v2}
\draw (\u) to (\v);
% signs
\node [left = 0 of v1] {$-$};
\node [below left = 0 and 0 of v2] {$+$};
\node [below right = 0 and -0.35 of v2] {$+$};
\node [above right = -0.1 and 0 of v2] {$+$};
\node [above =-0.1 of v3] {$-$};
\node [below right = -0.2 and 0 of v3] {$+$};
\node [above left = -0.1 and 0 of v4] {$-$};
\node [below left = -0.2 and 0 of v4] {$+$};
% vertex name
\node [above left= -0.2 and 0.05 of v2]{$s$};
\node [above right= -0.05 and -0.05 of v1]{$t$};
\end{tikzpicture}

(a) Original graph.
\end{minipage}
\begin{minipage}{0.3\textwidth}
\centering
\begin{tikzpicture}
\tikzset{vertex/.style={fill = white, draw, circle,thick}};
% vertices
\node[vertex] (v1) at (2,2) {};
\node[vertex] (v2) at (0,1.5) {};
\node[vertex] (v3) at (-1,-0.5) {};
\node[vertex] (v4) at (1,-0.5) {};
% links
\foreach \u / \v in {v1/v2,v2/v3,v3/v4,v4/v2}
\draw (\u) to (\v);
% signs
\node [left = 0 of v1] {$-$};
\node [below left = 0 and 0 of v2] {$+$};
\node [below right = 0 and -0.35 of v2] {$+$};
\node [above right = -0.1 and 0 of v2] {$+$};
\node [above =-0.1 of v3] {$-$};
\node [below right = -0.2 and 0 of v3] {$+$};
\node [above left = -0.1 and 0 of v4] {$-$};
\node [below left = -0.2 and 0 of v4] {$+$};
% additional 
\draw [bend right,line width = 1pt] (v2) to (v1);
\node [below right = -0.15 and 0 of v2] {$-$};
\node [below left = 0.1 and -0.1 of v1] {$+$}; 
% vertex name
\node [above left= -0.2 and 0.05 of v2]{$s$};
\node [above right= -0.05 and -0.05 of v1]{$t$};
\end{tikzpicture}

(b) One arc added.
\end{minipage}
\begin{minipage}{0.3\textwidth}
\centering
\begin{tikzpicture}
\tikzset{vertex/.style={fill = white, draw, circle,thick}};
% vertices
\node[vertex] (v1) at (2,2) {};
\node[vertex] (v2) at (0,1.5) {};
\node[vertex] (v3) at (-1,-0.5) {};
\node[vertex] (v4) at (1,-0.5) {};
% links
\foreach \u / \v in {v1/v2,v2/v3,v3/v4,v4/v2}
\draw (\u) to (\v);
% signs
\node [left = 0 of v1] {$-$};
\node [below left = 0 and 0 of v2] {$+$};
\node [below right = 0 and -0.35 of v2] {$+$};
\node [above right = -0.1 and 0 of v2] {$+$};
\node [above =-0.1 of v3] {$-$};
\node [below right = -0.2 and 0 of v3] {$+$};
\node [above left = -0.1 and 0 of v4] {$-$};
\node [below left = -0.2 and 0 of v4] {$+$};
% additional arc
\draw [bend right,line width = 1pt] (v2) to (v1);
\node [below right = -0.15 and 0 of v2] {$-$};
\node [below left = 0.1 and -0.1 of v1] {$+$}; 
% additional self loop
\draw [line width = 1pt] (v2) to [out=45,in=135,looseness=10] (v2);
\node [above = 0.4 of v2] {$-$}; 
% vertex name
\node [above left= -0.2 and 0.05 of v2]{$s$};
\node [above right= -0.05 and -0.05 of v1]{$t$};
\end{tikzpicture}

(c) Optimal solution.
\end{minipage}
\caption{Example: the size of optimal solution is $\lambda (G)+1$.
Bold lines represent the additional arcs in (b) and (c).}
\label{FigEx}
\end{figure}

Now we aim at obtaining the upper bound of the size of optimal solutions.
Actually, we can show the next theorem.
\begin{Thm}
Let $G=(V, A; \partial )$ be an acyclic bidirected graph.
Then the minimum number of $|A'|$ such that $G'=(V, A\cup A'; \partial' )$ is a strongly connected bidirected graph is $\lambda (G)$ or $\lambda (G)+1$.
\label{ThmArc}
\end{Thm}
Note that if the output of ADDITIONAL SIGNS($G$) contains at most one self-loop, then it is also an optimal solution of the problem of minimizing the number of additional arcs.
If the output of ADDITIONAL SIGNS($G$) contains more than 1 self-loops, however, we cannot guarantee the optimality for the problem on arcs in general.
We can construct a feasible solution of size $\lambda (G)$ or $\lambda (G)+1$ by the following algorithm.
\newline

\noindent \textbf{Algorithm} ADDITIONAL ARCS($G$)
\begin{description}
\item[Step 1] Let $A':=\emptyset$.
\item[Step 2] Let $u_1, u_2, \ldots , u_{L_1} \ (L_1:=k_1-|Q|)$ be the elements of $\left( \bigcup_{i=1}^{k_1}C_i^1 \right) \cap (S\cup T\cup Q')$.
If $L_1=\gamma =1$, add a self-loop at $u_1$ to $A'$ with a proper sign and go to Step 14.
If $L_1=\gamma >1$, add $\{ (u_1, u_i)\mid 2 \leq i \leq L_1\}$ to $A'$ with proper signs and go to Step 14.
\item[Step 3]
Let $\mathcal{C}=\left\{ C_1^2, C_2^2, \ldots , C_{k_2}^2, \ldots , C_1^K, C_2^K, \ldots , C_{k_K}^K \right\}$.
For each $C\in \mathcal{C}$, pick up two distinct elements of $C\cap (S\cup T)$ and label them as $l_i, r_i\ (i=1, 2, \ldots , |\mathcal{C}|)$.
Label the rest of the elements of $\bigcup_{C\in \mathcal{C}}C\cap (S\cup T)$ as $w_1, w_2, \ldots , w_{L_2}$ with $L_2 :=\sum_{i=3}^{K}(i-2)k_{i}$.
Add $\left\{ (u_i, w_i)\mid 1\leq i\leq \min\{ L_1, L_2 \} \right\}$ to $A'$ with proper signs.
\item[Step 4] 
Let $q_1, \ldots , q_{|Q|}$ be the elements of $Q$ and 
define $l_{| \mathcal{C} |+i}=r_{| \mathcal{C} |+i}=q_i$ for $i=1, 2, \ldots , |Q|$. 
\item[Step 5] If $L_2\leq L_1-2$, add $\left\{ (u_i, r_{|\mathcal{C}|+|Q|})\mid L_2 +1<i\leq L_1\right\}$, $\left\{ (r_i, l_{i+1}) \mid 1 \leq i<|\mathcal{C}|+|Q|\right\}$ and $(u_{L_2+1}, l_1)$ to $A'$ with proper signs and go to Step 14.
\item[Step 6] If $L_2=L_1-1$, add $\{ (r_i, l_{i+1})\mid 1\leq i<|\mathcal{C}|+|Q|\}$, $(u_{L_1}, l_1)$ and a self-loop at $r_{|\mathcal{C}|+|Q|}$ with proper signs and go to Step 14.
\item[Step 7] If $|Q|=1$, add a new vertex $q_2$ to $V$ and add $(q_1, q_2)$ with $(+, -)$ to $A'$.
Otherwise, add $(q_i, q_{i+1})$ to $A'$ with $(+, -)$ for $i=1, 2, \ldots , |Q|-1$.
\item[Step 8] Define a new bidirected graph $\hat{G}=(\hat{V}, \hat{A}; \hat{\partial})$ from the bidirected graph $G'=(V, A\cup A'; \partial' )$ as follows:
\[
\hat{V}:=\{ l_i, r_i\mid 1\leq i\leq |\mathcal{C}|\} \cup \{ w_j\mid L_1<j\leq L_2\} \cup \{ q_1, q_{\max \{ 2, |Q|\} }\} ,
\]
\[
\hat{A}:=\left\{ a=(u, v)\mbox{ with }(\pi _u, \pi _v)\mid \exists (\pi _u, \pi _v)\mbox{-path from }u\mbox{ to }v\mbox{ in }G', \{ u, v\} \subseteq \hat{V}\right\} .
\]
\item[Step 9] Construct a maximal matching $M=\{ m_1, m_2, \ldots , m_{|M|}\}$ ($m_i=(v_i^l, v_i^r)$) in the underlying graph of $\hat{G}$.
Add $B:=\{ (v_i^r, v_{i+1}^l)\mid 1\leq i\leq |M|\ (v_{|M|+1}^l:=v_1^l)\} $ to $A'$ with proper signs.
\item[Step 10] Let $p_1, p_2, \ldots , p_l$ be the vertices in $\hat{V}$ which are not the endpoints of any element of $M$.
Add $P:=\left\{ (p_{2i-1}, p_{2i})\mid 1\leq i\leq \lfloor l/2\rfloor \right\}$ with proper signs to $A'$.
If $l$ is odd, add a self-loop at $p_l$ to $A'$ with a proper sign.
\item[Step 11] Let $\tilde{G}$ be the output and $\alpha$ be the map in Step~5 of CONDENSE($G=(V, A\cup A'; \partial' )$).
\item[Step 12] Let $\tilde{v}$ be the only one element of $S(\tilde{G})\cup T(\tilde{G})\cup Q'(\tilde{G})\cup Q(\tilde{G})$.
If $\tilde{v}\in S(\tilde{G})$ (resp. $\tilde{v}\in T(\tilde{G})$), add a minus-loop (resp. a plus-loop) at an arbitrary element in $\alpha^{-1}(\tilde{v})$.
\item[Step 13] If $|Q|=1$ holds for the original input graph $G$, then remove the arc $(q_1, q_2)$ from $A'$.
\item[Step 14] Return $G'=(V, A\cup A'; \partial' )$.
\end{description}
Note that the above algorithm finds a feasible solution of size $\lambda (G)$ or $\lambda (G)+1$ in linear time (Table~\ref{TableApprox}).

Since the algorithm is the same as ADDITIONAL SIGNS($G$) when $L_2\leq L_1-1$, let us give a brief explanation on the case of $L_2\geq L_1-1$.
It is sufficient to show that $|S(\hat{G})\cup T(\hat{G})\cup Q'(\hat{G})\cup Q(\hat{G})|=1$ holds after Step 11 is executed.
By Step 10 of the algorithm, there is a proper cycle $C$ consisting of alternating sequence of $M$ and $P$ in $\hat{G}$.
The other elements in $S(\hat{G})\cup T(\hat{G})\cup Q'(\hat{G})\cup Q(\hat{G})$ are included in some improper cycle the root of which is included in $C$.
Therefore, there exists only one element of $S(\hat{G})\cup T(\hat{G})\cup Q'(\hat{G})\cup Q(\hat{G})$ after Step 11.
If the graph is not strongly connected, the all vertices have become strongly connected by adding a self-loop with the proper sign.

The cardinality of the solution is $\lambda (G)+1$ only if $\tilde{v}\in S(\tilde{G})\cup T(\tilde{G})$.
The above algorithm runs in linear time, thus the total algorithm runs in linear time for a general input bidirected graph.
\begin{Thm}
For Problem~\ref{ProbGeneral} with $F(A', \partial' )=|A'|$, a feasible solution with $|A'|=\lambda (G)$ or $\lambda (G)+1$ can be found in linear time.
\end{Thm}
\begin{table}
\caption{The relation between the output of the algorithm and the optimal solution.}
\label{TableApprox}
\centering
\begin{tabular}{|c||c|c|}\hline
Optimum$\backslash$Output&$\lambda$&$\lambda +1$\\ \hline \hline
$\lambda$&optimal&approximate\\ \hline
$\lambda +1$&$\not\exists$&optimal\\ \hline
\end{tabular}
\end{table}
Actually, there exists an example such that our algorithm returns the approximate solution (Figure~\ref{FigApproxEx}).
The original graph has five vertices and six $(+, +)$-arcs (Figure~\ref{FigApproxEx} (a)).
If one applies our algorithm to this graph, a solution of four arcs is obtained (Figure~\ref{FigApproxEx} (b)).
However, there exists a solution of three arcs (Figure~\ref{FigApproxEx} (c)).
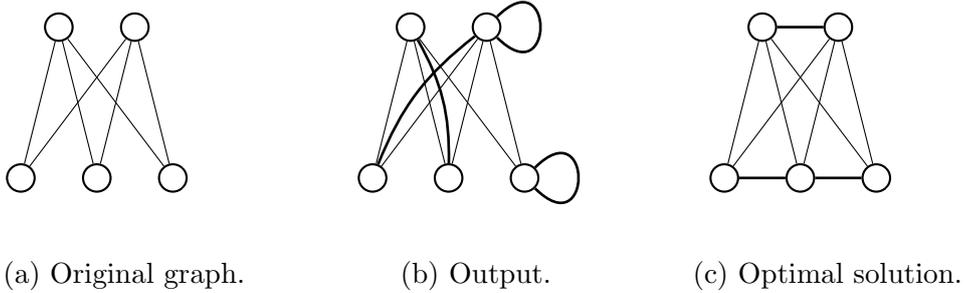
\begin{figure}[H]
	\centering
	\begin{minipage}{0.3\textwidth}
		\centering
		\begin{tikzpicture}
		\tikzset{vertex/.style={fill = white, draw, circle,thick}};
		% vertices
		\node[vertex] (v1) at (-0.5,1.5) {};
		\node[vertex] (v2) at (0.5,1.5) {};
		\node[vertex] (w1) at (-1,-0.5) {};
		\node[vertex] (w2) at (0,-0.5) {};
		\node[vertex] (w3) at (1,-0.5) {};
		% links
		\foreach \u / \v in {v1/w1,v1/w2,v1/w3,v2/w1,v2/w2,v2/w3}
		\draw (\u) to (\v);
		% dummy
		\draw [line width = 1pt,white] (w3) to [out=-45,in=45,looseness=10] (w3);
		\draw [line width = 1pt,white] (v2) to [out=-45,in=45,looseness=10] (v2);
		\end{tikzpicture}

		(a) Original graph.
	\end{minipage}
	\begin{minipage}{0.3\textwidth}
		\centering
		\begin{tikzpicture}
		\tikzset{vertex/.style={fill = white, draw, circle,thick}};
		% vertices
		\node[vertex] (v1) at (-0.5,1.5) {};
		\node[vertex] (v2) at (0.5,1.5) {};
		\node[vertex] (w1) at (-1,-0.5) {};
		\node[vertex] (w2) at (0,-0.5) {};
		\node[vertex] (w3) at (1,-0.5) {};
		% links
		\foreach \u / \v in {v1/w1,v1/w2,v1/w3,v2/w1,v2/w2,v2/w3}
		\draw (\u) to (\v);
		% additional links
		\draw [bend left=15,line width = 1pt] (v1) to (w2);
		\draw [bend right=15,line width = 1pt] (v2) to (w1);
		\draw [line width = 1pt] (w3) to [out=-45,in=45,looseness=10] (w3);
		\draw [line width = 1pt] (v2) to [out=-45,in=45,looseness=10] (v2);
		\end{tikzpicture}
		
		(b) Output.
	\end{minipage}
	\begin{minipage}{0.3\textwidth}
		\centering
		\begin{tikzpicture}
		\tikzset{vertex/.style={fill = white, draw, circle,thick}};
		% vertices
		\node[vertex] (v1) at (-0.5,1.5) {};
		\node[vertex] (v2) at (0.5,1.5) {};
		\node[vertex] (w1) at (-1,-0.5) {};
		\node[vertex] (w2) at (0,-0.5) {};
		\node[vertex] (w3) at (1,-0.5) {};
		% links
		\foreach \u / \v in {v1/w1,v1/w2,v1/w3,v2/w1,v2/w2,v2/w3}
		\draw (\u) to (\v);
		% additional links
		\draw[line width = 1pt] (v1) to (v2);
		\draw[line width = 1pt] (w1) to (w2);
		\draw[line width = 1pt] (w2) to (w3);
	    % dummy
		\draw [line width = 1pt,white] (w3) to [out=-45,in=45,looseness=10] (w3);
		\draw [line width = 1pt,white] (v2) to [out=-45,in=45,looseness=10] (v2);
		\end{tikzpicture}
		
		(c) Optimal solution.
	\end{minipage}
	\caption{Example for which the algorithm returns an approximate solution.
Signs on arcs are omitted.
(a) All arcs are $(+, +)$-arcs.
(b) Four bold arcs are additional ones.
Links are $(-, -)$-arcs and self-loops are with the minus sign.
(c) Three bold arcs are additional $(-, -)$-arcs.}
	\label{FigApproxEx}
\end{figure}

%%%%%%%%%%%%%%%%%%%%%%%%%%%%%%%%%%%%%%%%%%%%%%%%%%%%%%%%%%%%
\section{Concluding Remarks}
In this paper, we propose two types of problems to make a given bidirected graph strongly connected.
The first one aims at minimizing the number of additional signs on arcs and the second one aims at minimizing the number of additional arcs.
We give a linear-time algorithm to find an optimal solution for the former problem and a linear-time algorithm to find a feasible solution which can have one more arc than an optimal solution for the latter problem.

As future works, the following problem on minimization on arcs can be considered.
\begin{Prob}
Let $G=(V, A; \partial )$ be a bidirected graph.
Decide whether the minimum number of additional arcs to make $G$ strongly connected is $\lambda (G)$ or $\lambda (G)+1$.
\end{Prob}

Connectivity augmentation problems on bidirected graphs can also be considered,
e.g., arc-connectivity augmentation.
Let $G$ be $k$-arc-connected if $G'=(V, A\setminus A^{\circ}; \partial |(A\setminus A^{\circ}))$ is strongly connected for all $A^{\circ}\subseteq A$ with $|A^{\circ}|=k-1$.
\begin{Prob}
Let $G=(V, A; \partial )$ be a bidirected graph and $k$ be a positive integer.
Find additional arcs $A'$ and a boundary operator $\partial' : A\cup A'\to 3^V\ (\partial' a=\partial a\ (\forall a\in A))$ minimizing $F(A', \partial' )$ such that $G$ is $k$-arc-connected.
\end{Prob}
In a similar way, the definition of $k$-vertex-connectivity and the connectivity augmentation problem on it can be introduced.

Also, the generalization of the problem to make a given undirected graph connected or that to make a given directed graph strongly connected can be considered.
Although bidirected graphs can be seen as the common generalization of undirected graphs and directed graphs, the problems in this paper are not the generalization of these classical problems because there is no restriction on additional arcs.
For the case of directed graphs, the problem can be formulated as follows: 
\begin{Prob}
Let $G=(V, A; \partial )$ be a bidirected graph.
Find additional arcs $A'$ and a boundary operator $\partial' : A\cup A'\to 3^V\ (\partial' a=\partial a\ (\forall a\in A))$ minimizing $|A'|$ such that $G':=(V, A\cup A'; \partial' )$ is a strongly connected bidirected graph and $\left| \partial'^+a'\right| =\left| \partial'^-a'\right| =1$ for all $a'\in A'$.
\end{Prob}
\section*{Acknowledgments}
Both authors are supported by JSPS Research Fellowship for Young Scientists.
The research of the first author was supported by Grant-in-Aid for JSPS Research Fellow Grant Number 16J06879.

\end{document}